\documentclass[11pt]{article}
\usepackage[utf8]{inputenc}

\usepackage{amsmath,amsthm,amssymb,xcolor,fullpage,hyperref,bbm}
\usepackage[capitalize,noabbrev]{cleveref}
\usepackage{algorithm, algorithmicx, algpseudocode}
\usepackage{xspace}
\usepackage{graphicx}
\usepackage{thm-restate}
\usepackage{stackengine}

\newcommand{\tildecircume}{\stackon[2pt]{\^{e}}{\~{}}}

\newtheorem{definition}{Definition}
\newtheorem{lemma}{Lemma}

\newtheorem{theorem}{Theorem}

\newtheorem{claim}{Claim}
\newtheorem{remark}{Remark}
\newtheorem{corollary}{Corollary}

\newcommand{\eps}{\epsilon}

\newcommand{\N}{\mathbb{N}}

\DeclareMathOperator*{\B}{\mathbf{B}}
\DeclareMathOperator*{\w}{\mathbf{w}}
\DeclareMathOperator*{\preds}{\mathbf{\hat{y}}}
\DeclareMathOperator*{\losses}{\boldsymbol{\ell}}

\newcommand{\mainalg}{\texttt{HierarchicalBaseline}}
\newcommand{\bucketsalg}{\texttt{UpdateBuckets}}
\newcommand{\weightsalg}{\texttt{UpdateWeights}}
\newcommand{\predictalg}{\texttt{GetPredictions}}

\newcommand{\stay}{stay\xspace}
\newcommand{\evict}{evict\xspace}

\newcommand{\sandeep}[1]{\textcolor{orange}{[sandeep\@ifnotempty{#1}{: #1}]}}

\title{Improved Space Bounds for Learning with Experts}
\author{Anders Aamand \\ MIT \\ \texttt{aamand@mit.edu} \and Justin Y.\ Chen \\ MIT \\ \texttt{justc@mit.edu}  \and Huy L\^{e} Nguy\tildecircume n\\ Northeastern University \\ \texttt{hu.nguyen@northeastern.edu} \and Sandeep Silwal\\ MIT  \\ \texttt{silwal@mit.edu}}
\date{}

\begin{document}
\maketitle
\begin{abstract}
We give improved tradeoffs between space and regret for the online learning with expert advice problem over $T$ days with $n$ experts. Given a space budget of $n^{\delta}$ for $\delta \in (0,1)$, we provide an algorithm achieving regret $\tilde{O}(n^2 T^{1/(1+\delta)})$, improving upon the regret bound $\tilde{O}(n^2 T^{2/(2+\delta)})$ in the recent work of \cite{peng2023onlinesublinear}. The improvement is particularly salient in the regime $\delta \rightarrow 1$ where the regret of our algorithm approaches $\tilde{O}_n(\sqrt{T})$, matching the $T$ dependence in the standard online setting without space restrictions.

\end{abstract}

\section{Introduction}
Understanding the performance of learning algorithms under information constraints is a fundamental research direction in machine learning. While performance notions such as regret in online learning have been well explored, a recent line of work explores \emph{additional} constraints in learning, with a particular emphasis on \emph{limited memory} \cite{shamir14, woodworthS19, marsdenSSV22} (see also Section~\ref{sec:related_work}).

In this paper, we focus on the online learning with experts problem, a general framework for sequential decision making, with memory constraints. In the online learning with experts problem, an algorithm must make predictions about the outcome of an event for $T$ consecutive days based on the predictions of $n$ experts. The predictions of the algorithm at a time $t \le T$ can only depend on the information it has received in the previous days as well as the predictions of the experts for day $t$. After predictions are made, the true outcome is revealed and the algorithm and all experts receive some loss (likely depending on the accuracy of their predictions). In addition to the fact that the online experts problem has found numerous algorithmic applications \cite{aroraHK12}, studying the problem with memory constraints is especially interesting in light of the fact that existing algorithms explicitly track the cumulative loss of every expert and follow the advice of a leading expert, which requires $\Omega(n)$ memory. 

Motivated by this lack of understanding, the online learning with experts problem with memory constraints was recently introduced in \cite{srinivas2022memoryexperts}, which studied the case where the losses of the experts form an i.i.d. sequence or where the loss of the best expert is bounded. The follow up work of \cite{peng2023onlinesublinear} removed these assumptions and obtained an algorithm which achieves $\tilde{O}(T^{\frac{2}{2+\delta}})$\footnote{$\tilde{O}$ hides polylogarithmic factors in $T$, and for simplicity we ignore additional $\text{poly}(n)$ overhead factors in the regret bounds in this part of the introduction.} regret
using memory $n^{\delta}$ in a general setting; see Section \ref{sec:related_work} for a more detailed comparison. Intuitively, the result of \cite{peng2023onlinesublinear} suggests that only keeping track of an (evolving) set of $n^{\delta}$  experts at any fixed day is sufficient to achieve sublinear regret. However, their work leaves open a natural question in the case where the space budget approaches near linear space. In this regime, it is natural to guess that regret $\tilde{O}(\sqrt{T})$ is achievable, namely the regret bound achieved by the standard multiplicative weights update (MWU) algorithm (among many others \cite{littlestoneW89,kalaiV03, aroraHK12, hazan16}) which uses $O(n)$ space. However, the algorithm of \cite{peng2023onlinesublinear} only achieves regret $\tilde{O}(T^{2/3})$ in this regime.

We close this gap in understanding by providing an algorithm with approximately $\tilde{O}(T^{\frac{1}{1+\delta}})$ regret using $n^{\delta}$ memory, thus obtaining regret $\tilde{O}(\sqrt{T})$ in the near linear memory regime.

\subsection{Our Results}
We give a brief overview of the problem setting to state our results, deferring the full details to Section~\ref{sec:prelim}. In the online experts problem, on each day over a sequence of $T$ days, we are to play one of $n$ experts. 
After playing an expert $i_t\in [n]$ on day $t$, a loss vector $\ell_t\in [0,1]^n$ is revealed and we receive the loss $\ell_t(i_t)$. 
In this paper, we assume that the loss sequence of each expert is picked by an oblivious adversary. Our goal is to minimize the standard notion of regret in online learning, defined as the total loss of the predictions made by our algorithm in comparison to the total loss of the best expert in hindsight: $\text{Regret} = \sum_{t \in [T]}
\ell_t(i_t) - \min_{i \in [n]} \sum_{t \in [T]}\ell_t(i)$. Our main result is the following.

\begin{restatable}{theorem}{mainthm}
\label{thm:main}
For any $\delta\in (0,1)$, there exists an algorithm for the online experts problem over $T=n^{O(1)}$ days which uses space $\tilde O(n^\delta)$ and achieves regret $\tilde O(n^2T^{\frac{1}{1+\delta}})$ with probability $1-1/\text{poly}(T)$.
\end{restatable}

In contrast, \cite{peng2023onlinesublinear} gave an algorithm achieving regret $\tilde{O}(n^2T^{\frac{2}{2+\delta}})$ for a comparable space budget. Concretely, we obtain a smaller exponent of $T$ in the regret bound for all values of $\delta$; see Figure \ref{fig:regret}. And in the regime of near linear space where $\delta \rightarrow 1$, we achieve regret $\tilde{O}_n(\sqrt{T})$, matching the guarantees of traditional online algorithms (such as MWU) which explicitly track the performance history of all $n$ experts. In comparison, \cite{peng2023onlinesublinear} achieve regret $\tilde{O}_n(T^{2/3})$ in this regime.

\section{Technical Overview}\label{sec:technical_overview}
Within this overview, we ignore logarithmic factors for simplicity. We first describe a high level overview of the algorithm in \cite{peng2023onlinesublinear} and then describe the new ideas in this work to achieve the improved bound.
The upper bound in \cite{peng2023onlinesublinear} comprises two parts: a baseline algorithm and a bootstrapping procedure. The baseline algorithm achieves bounded regret in terms of $T$, $n$, and a space parameter $m=n^\delta$ for $\delta \in (0,1)$. For some \emph{specific} setting of $T$ in terms of $n$ and $m$, this baseline algorithm achieves regret of $T^{\frac{2}{2+\delta}}$. The bootstrapping procedure allows for this same bound to be extended for larger values of $T$. Our main contribution is to develop an improved baseline algorithm which achieves a better regret tradeoff in terms of $T$, $n$, and $m$. Additionally, we simplify and shorten the analysis of the bootstrapping procedure and present a lemma which essentially combines any two algorithms $\mathcal{A}_1$ and $\mathcal{A}_2$ with regret $R_1$ and $R_2$ over times $T_1$ and $T_2$ respectively into an algorithm with regret $R_1R_2$ over $T_1T_2$ days.

\subsection{Baseline algorithm of \cite{peng2023onlinesublinear}}

The core idea of the baseline algorithm from \cite{peng2023onlinesublinear} is to split the time $T$ into a sequence of consecutive blocks of size $T_0 < T$. Within each block, a pool of $m + \sqrt{m}$ experts is maintained. The prediction of the algorithm at any day is the output of MWU run over the experts in the current block. The weights are reset every time we move to the next block every $T_0$ timesteps. The blocks are initially set to a random sample of $m$ experts. For each expert $e$ in the pool, the algorithm tracks their arrival time as well as their average loss since any other expert $e'$ in the pool arrived. This requires memory quadratic in the number of unique arrival times of experts in the pool.

The algorithm also keeps experts which have been in the pool for a long time and have been performing relatively well. Say $t(e)$ is the time at which an expert $e$ arrived in the pool and let $\eps = 1/\sqrt{m}$. At the end of every block, if there exists a pair of experts $e, e'$ such that $t(e') < t(e)$ and the average loss of $e'$ since $t(e)$ is at most $(1+\eps)$ that of $e$, we say that $e'$ ``dominates'' $e$ and remove $e$ from the pool. After removing all dominated experts, the algorithm randomly samples and adds $m$ new experts to the pool.

\paragraph{Memory}
The number of experts in the pool at any given point is not explicitly bounded by the algorithm but is limited to $m$ newly sampled experts and $\sqrt{m}$ older experts due to the domination rule. This is proved formally using a potential argument in Section 3.2 of \cite{peng2023onlinesublinear}. Our first improvement is a domination strategy that provides a \emph{quadratic} improvement on the number of older experts stored for the same amount of memory and is much simpler to analyze. We will visit it shortly in Section \ref{sec:improved_domination}.

\paragraph{Regret}
The regret bound achieved from the baseline algorithm of \cite{peng2023onlinesublinear} is given by the following equation
\begin{equation}\label{eq:priorwork-regret}
    \frac{T}{\sqrt{m}} + \frac{T}{\sqrt{T_0}} + \frac{nT_0}{m}.
\end{equation}

In order to bound the regret of the baseline algorithm, the prior work of \cite{peng2023onlinesublinear} introduces the insightful concept of \stay and \evict-blocks (although using different terminology). Before sampling at the start of each block, we ask the hypothetical question: \emph{given the current pool of experts, if the best expert $e^*$ were sampled at this point, would it stay for the rest of time}? Note that the only reason $e^*$ would not stay for the rest of time is if there exists some expert already in the pool which will, at some point, dominate $e^*$. If the answer to this question is ``yes'', the block is a \stay-block. Otherwise, we call it an \evict-block.

For any \stay-block, if we do indeed sample the best expert, it will stay forever, and then the only regret we will pay will be the regret due to running MWU within each bucket. There are $T/T_0$ buckets with $\sqrt{T_0}$ regret in each, contributing to the second term of \cref{eq:priorwork-regret}.
If we do not sample the best expert, a \stay-block can cost up to $T_0$ regret as all of the experts in the pool may perform much worse than the best expert. However, as have $m$ chances to sample the best expert out of $n$ total experts, we will likely only see $n/m$ \stay-blocks without sampling the best expert. This contributes to the third term of \cref{eq:priorwork-regret}.

For an \evict-block, we will evict the best expert even if we sample them, but because of this we know that the expert $e'$ that evicts $e^*$ must be competitive with $e^*$ over some interval. In particular, let $t'$ be the time at which $e'$ dominates and evicts $e^*$. In order for this to happen, $e'$ must remain in the pool until $t'$ and must have average loss over the interval up to $t'$ which is at most a $(1+\eps) = (1 + 1/\sqrt{m})$ factor worse than the best expert. During these intervals starting at \evict-blocks, we pay average regret of $1/\sqrt{m} + 1/\sqrt{T_0}$ (the second term coming from the overhead of using MWU) which contributes to the first and second terms of \cref{eq:priorwork-regret}.

The final regret bound in terms of $T$ can be attained by setting $T_0 = m$ and $T = n\sqrt{m} = n^{1 + \delta/2}$. Then, the total regret is $n = T^{\frac{1}{1 + \delta/2}} = T^{\frac{2}{2 + \delta}}$.

\subsection{Improved domination strategy}\label{sec:improved_domination}
Our first contribution is an improved and simpler definition of dominance. For a pair of experts $e, e'$ in the pool, we say that $e'$ dominates $e$ if $t(e') < t(e)$ and the loss of $e'$ since $t(e')$ is at most a $(1+\eps)$ factor greater than the loss of $e$ since $t(e)$. The key difference with prior work is that the loss we use to compare two experts is just the loss of those experts \emph{since their own arrival} in the pool rather than their loss since the arrival of the later expert in the pair. Note that because of this we only need to store \emph{linear} (rather than quadratic) information in the number of experts in the pool. So we can set $\eps = 1/m$ rather than $1/\sqrt{m}$. This change also greatly simplifies the memory analysis which we outline below.

\paragraph{Memory}
At the end of each block, experts are evicted and $m$ new experts are sampled. We will argue that after this procedure, at most $2m$ experts are being tracked by the algorithm. Consider a pair of experts $e, e'$ that still remain in the pool after the eviction step. Without loss of generality, assume that $e'$ arrived before $e$. 
Then, the loss of $e$ since it arrived must be less than a $1/(1+\eps)$ factor of the loss of $e'$ since it arrived. As this holds across all pairs and losses must be in $[0,T]$, there can be at most $\log_{1/(1+\eps)}(T) = \tilde{O}(1/\eps)$. With the proper setting of $\eps$, this means that there can be at most $m$ experts after eviction and therefore at most $2m$ experts being tracked by the algorithm at any time. 

\paragraph{Regret}
The regret of the baseline algorithm with this new domination rule is
\begin{equation}\label{eq:onelevel-regret}
    \frac{T}{m} + \frac{T}{\sqrt{T_0}} + \frac{nT_0}{m}.
\end{equation}

We account for regret using the key idea of \stay and \evict-blocks from \cite{peng2023onlinesublinear}. The key difference from \cref{eq:priorwork-regret} to our regret bound in \cref{eq:onelevel-regret} is the $\eps$ used for dominance is quadratically smaller in our construction, leading to $T/m$ rather than a $T/\sqrt{m}$ term.

The final regret bound in terms of $T$ can be gained by setting $T_0 = m^2$ and $T = nm^2 = n^{1 + 2\delta}$. Then, the total regret is $nm = n^{1 + \delta} = T^{\frac{1 + \delta}{1 + 2\delta}}$. This is a strict improvement over the bound from \cite{peng2023onlinesublinear} (see the first gray curve in Figure \ref{fig:regret}), however, with linear memory, it still gives the undesirable bound of $T^{2/3}$ regret.

\subsection{Improved hierarchical baseline}
A fundamental bottleneck in the above algorithm is the tradeoff between sampling new experts, which corresponds to the $nT_0/m$ term in the regret in \cref{eq:onelevel-regret}, and keeping any sampled expert for a long time horizon to sufficiently reap the benefits of the MWU, which corresponds to the term $T/\sqrt{T_0}$ in \cref{eq:onelevel-regret}. The term $nT_0/m$ incentives setting $T_0$ to be a small quantity so that we essentially sample at a higher frequency, leading to an increased likelihood of sampling the best expert in one of the \stay-blocks mentioned above. On the other hand, the latter term $T/\sqrt{T_0}$ incentives keeping any expert around for a larger value of $T_0$. The above algorithm must balance between these two conflicting goals. An idealistic goal is to obtain the `best of both worlds' by setting $T_0$ to be small in $nT_0/m$, and thereby `exploring' many experts, while simultaneously setting $T_0$ to be small in $T/\sqrt{T_0}$ and `exploiting' the experts that we are currently tracking. Our improvement is inspired by attempting to implement such a hypothetical plan of action.

In order to further improve the algorithm, we use a hierarchy of $k$ blocks of different sizes: we divide time $T$ into blocks of $T_{k-1}$ and divide those blocks further into blocks of size $T_{k-2}$ and so on down to the smallest blocks of size $T_0$. By passing along good experts between hierarchies and viewing each hierarchy as a meta expert, we can effectively sample experts at a higher rate due to shorter hierarchies while retaining the benefits of MWU from longer spanning hierarchies, thus achieving the `best of both worlds'.

In this overview, we give a detailed description of the $k=2$ case to showcase why using multiple blocks is helpful (our algorithm will use $k \approx \log\log T$).

In \cref{eq:onelevel-regret}, the block size $T_0$ appears in the second and third terms. A larger block size reduces regret in the second term because we run MWU longer between resets and the average regret of MWU is inversely proportional to the square root of the span of time over which it is run. On the other hand, a smaller block size reduces regret that we pay in \stay-blocks where we do not sample the best expert. In other words, small blocks mean we sample more often. The goal of using two different block sizes will be to get the best-of-both-worlds by charging regret to large \evict-blocks and small \stay-blocks.

\begin{figure}[ht]
    \centering
    \includegraphics[width=\textwidth]{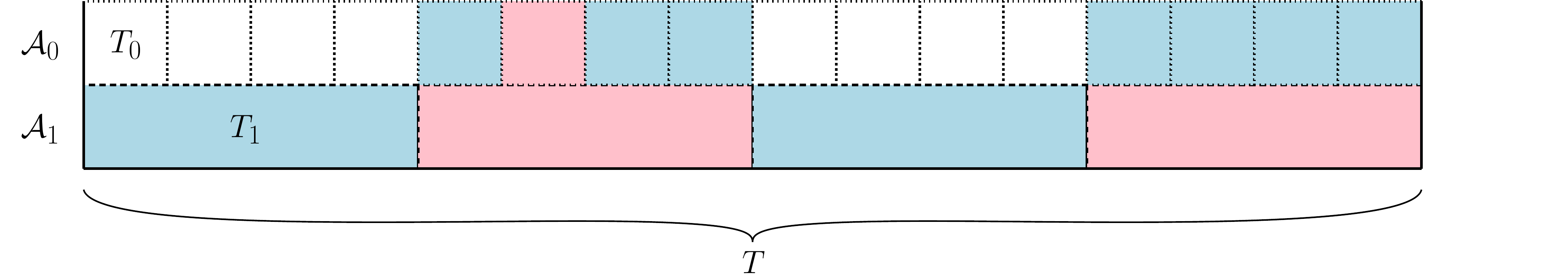}
    \caption{The figure depicts our hierarchical baseline scheme in the special case of only two hierarchies. The blue blocks denote evict-blocks and red blocks denote stay-blocks.}
    \label{fig:toydiagram}
\end{figure}

Let $T_0 < T_1 < T$ be the block sizes with $T_0$ dividing $T_1$ and $T_1$ dividing $T$. We refer to Figure \ref{fig:toydiagram} for reference. Let $\mathcal{A}_0$ and $\mathcal{A}_1$ be the outputs of the one level baseline algorithm from the previous section using $T_0$ and $T_1$, respectively. The hierarchical baseline will be the output of MWU run between $\mathcal{A}_0$ and $\mathcal{A}_1$, reset every $T_1$ block. For each $T_1$ block, we can then bound the regret as the better of $\mathcal{A}_0$ and $\mathcal{A}_1$ plus an additional $\sqrt{T_1}$ term from MWU. For every $T_1$ \evict-block, we can bound the regret as in the first two terms of \cref{eq:onelevel-regret} but using block size $T_1$. For every $T_1$ \stay-block, we will consider the regret of $\mathcal{A}_0$.

Every $T_0$ \evict-block contributes two terms to the regret: one for the loss of the evicting expert with respect to the best expert and one for running MWU. The first term contributes at most an additive $T/m$ to the total regret. The second term contributes $\sqrt{T_0}$ for each \evict-block we count. However, as we only count $T_0$ \evict-blocks in $T_1$ \stay-blocks and after $n/m$ $T_1$ \stay-blocks we will sample and retain the best expert, this contributes a total regret of $\frac{n}{m} \left(\frac{T_1}{T_0}\right) \sqrt{T_0}$.

Finally, we consider the $T_0$ \stay-blocks in $T_1$ \stay-blocks. In total, there can be $\frac{n}{m}$ $T_0$ \stay-blocks before we sample the best expert and it stays forever. However, even if we sampled the best expert in the $T_0$ blocks, we have to pay the $\sqrt{T_0}$ MWU cost every block. Therefore, we make a key algorithmic change: \emph{at the end of every $T_1$ block, we make a copy in $\mathcal{A}_1$ of every expert in $\mathcal{A}_0$ after eviction and we enforce that those copied experts in $\mathcal{A}_0$ cannot be evicted until the next $T_1$ block}. This increases our memory to at most $k^3m$ in total as each block maintains at most $km$ experts after eviction and receives at therefore receives at most $km$ experts from each of smaller blocks. Furthermore, this guarantees that if we ever sample the best expert into a $T_0$ \stay-block, it will be in every following $T_1$ block. Thus, the total contribution to regret of $T_0$ \stay-blocks in $T_1$ \stay-blocks is $\frac{n}{m} T_0 + \frac{T_1}{T_0}$.
In total, we can bound the regret of this two-level algorithm as
\begin{equation}\label{eq:twolevel-regret}
    \frac{T}{m} + \frac{T}{\sqrt{T_1}} + \frac{n}{m}\left(\frac{T_1}{\sqrt{T_0}} + T_0\right).
\end{equation}

Setting $T_0 = m^{4/3}$, $T_1 = m^2$, and $T = nm^{4/3} = n^{1 + 4\delta/3}$ yields regret $nm^{1/3} = n^{1 + \delta/3} = T^{\frac{1 + \delta/3}{1 + 4\delta/3}} = T^{\frac{3 + \delta}{3 + 4\delta}}$. This bound interpolates between regret $T$ and $T^{4/7}$ as $\delta$ goes from $0$ to $1$, improving upon the one-level algorithm across the board and in particular in the linear memory regime (see the second gray curve in Figure \ref{fig:regret}). The full hierarchical baseline algorithm in \cref{sec:algo} extends this hierarchical idea to $k = O(\log\log T)$ to achieve regret $\tilde{O}(T^{\frac{1}{1+\delta}})$ for $T = nm$.

We emphasize that it is not sufficient to simply run the algorithm of \cite{peng2023onlinesublinear} recursively on top of every $T_1$ block (see Figure \ref{fig:toydiagram}) to achieve $T^{\frac{1}{1+\delta}}$ regret. It is crucial to pass down experts from a higher level to a lower level to replace the `large' regret term $\frac{nT_1}m$ which is incurred by stay-blocks in $\mathcal{A}_1$ to the `small' regret term $\frac{nT_0}m$ incurred by stay-blocks in $\mathcal{A}_0$. However, we cannot only pass down experts from the higher level $\mathcal{A}_0$ to $\mathcal{A}_1$ as this creates a subtle issue: the time intervals where a good expert in $\mathcal{A}_0$ is competitive with $e^*$ can possibly overlap with the time interval when an expert in $\mathcal{A}_1$ is competitive with $e^*$. This creates the dilemma where we cannot guarantee that a competitive expert exists in the larger time interval which spans the \emph{union} of the time intervals when the two experts are competitive with $e^*$. To see this, consider the following toy example in the case $T = 5$: suppose expert $e_1$ has loss sequence $[1,0,0,0,1]$, $e_2$ has loss sequence $[0,0,0,1,1]$, $e^*$ has the loss sequence $[0,0,1,0,0]$, and we can only pick one expert to follow in these $5$ days. Then $e_1$ is competitive with $e^*$ in days $2$ through $5$ (in the sense that it receives the same total loss as $e^*$ during these days) and $e_2$ is competitive with $e^*$ in days $1$ through $4$, but no expert is competitive with $e^*$ in the entire time interval. 

We avoid this issue by enforcing that, in a small block, the experts which are passed to larger blocks cannot be evicted until the end of those larger blocks. This means that if the small block is an \evict-block, the evicting expert must be competitive with $e^*$ over a time period which is the union of a few small blocks and one or more larger blocks.  

\begin{figure}[ht]
    \centering
    \includegraphics[width=0.6\textwidth]{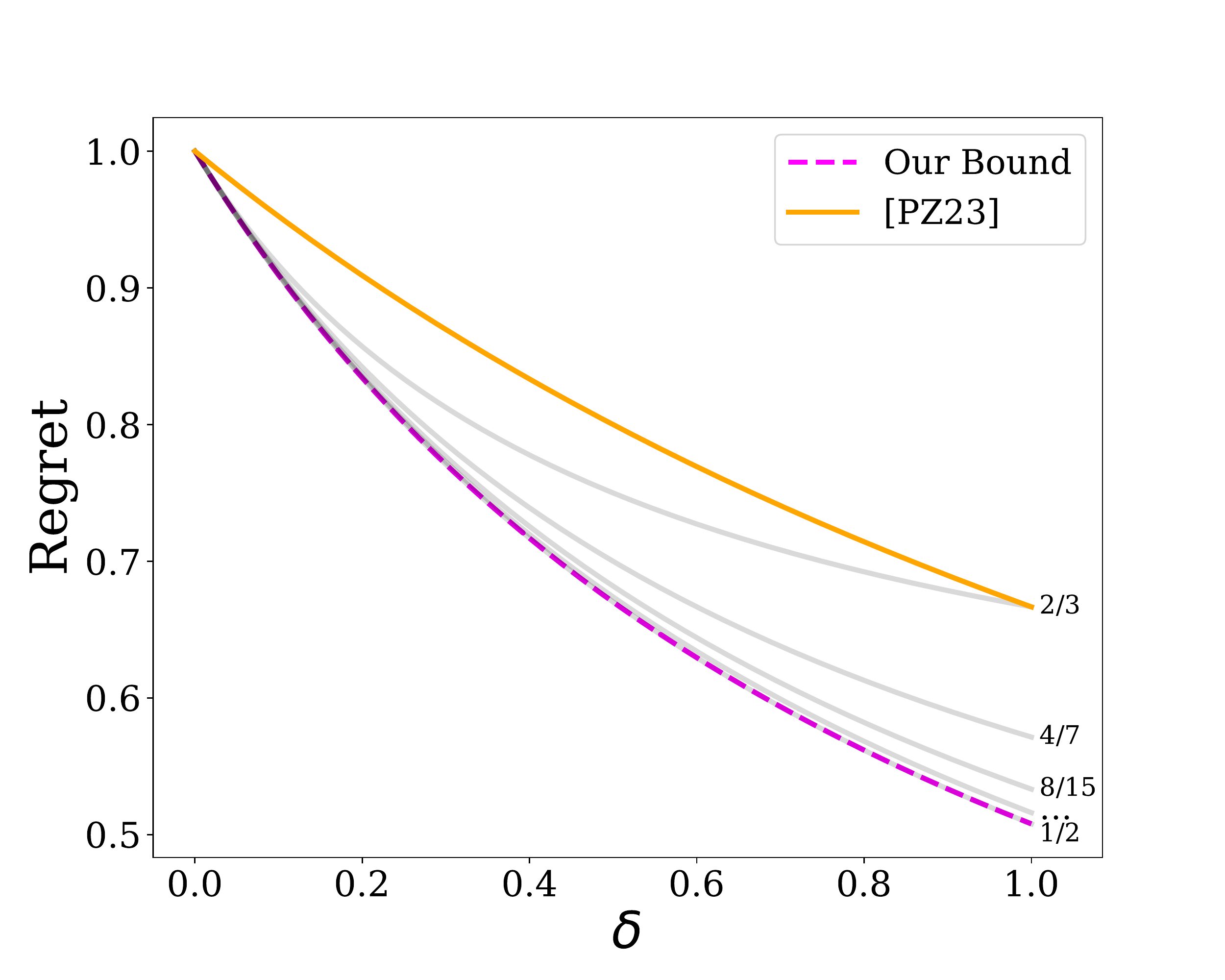}
    \caption{The $x$-axis represents the exponent of $n$ in the space budget and the $y$ axis represents the exponent of $T$ in the regret. The orange curve shows the plot $f(\delta) = 2/(2+\delta)$, representing the main result of \cite{peng2023onlinesublinear}. The gray curves represent our regret bounds using a hierarchy of $k$ blocks of different sizes for $k=1, 2, 3, \ldots$. Each setting of $k$ improves the memory-regret tradeoff, as indicated by the fact that in the near linear space regime, we obtain exponents of $2/3, 4/7, 8/15, \ldots$ as $k$ increases. Our final choice of $k = O(\log \log T)$ leads to the exponent of $1/(1+\delta)$ for $T$ given in Theorem \ref{thm:main}, and is represented by the dashed magenta curve.}
    \label{fig:regret}
\end{figure}

\subsection{Bootstrapping for larger $T$}
The results stated above hold for a particular choice of $T$. In actuality, $T$ is part of the input and cannot be controlled by the algorithm. The authors of \cite{peng2023onlinesublinear} introduce a recursive width reduction technique which bootstraps their baseline algorithm to hold for larger values of $T$. Width reduction refers to the fact that the procedure uses the baseline algorithm to reduce the range of losses.
We give a more general version of this result (stated below) along with a simpler analysis in \cref{sec:boostrap}.

\begin{restatable}{lemma}{bootstrap}\label{lem:bootstrap}
Suppose that there are two algorithms ALG1 and ALG2 for the expert
problem with $n$ experts with daily loss range $1$ (the difference
between the maximum and the minimum loss) with the following parameters:
ALG1 is over $T_{1}$ days and has regret at most $R_{1}$, ALG2 is
over $T_{2}$ days and has regret at most $R_{2}$, each with probability
$1-\delta$. Furthermore, suppose both algorithms have space complexity $m$. Then there is an algorithm ALG' over $T_{1}\cdot T_{2}$ days and
regret at most $R_{1}\cdot R_{2}+O\left(T_{1}\cdot\sqrt{T_{2}\ln\left(n/\delta\right)}\right)$
with probability $1-\delta\left(nT_{1}+1\right)$. Furthermore, ALG' has space complexity at most $O(m)$.
\end{restatable}

\section{Related Works}\label{sec:related_work}
\paragraph{Comparison with \cite{srinivas2022memoryexperts} and \cite{peng2023onlinesublinear}}
These works are the most relevant to us. \cite{srinivas2022memoryexperts} initiated the study of the online learning with experts problem with sublinear space. They showed that in the case where the loss sequence of each expert is i.i.d., $O(\sqrt{T} \cdot n^{(1-\delta)/2})$ regret is achievable with $n^{\delta}$ space. Their algorithm design is intricately tied to the i.i.d. assumption. A matching lower bound for the problem was also given in their paper (this lower bound also applies to our adversarial setting). We note that a qualitatively similar lower bound for the i.i.d. setting with limited memory was also proven in \cite{shamir14}. They also studied the (standard) worst-case model of an oblivious adversary and showed that sublinear regret is possible under the strong assumption that the best expert has \emph{sublinear} loss.

\cite{peng2023onlinesublinear} show that neither the i.i.d. assumption nor the assumption that the best expert achieves sublinear loss across all $T$ days are needed to achieve sublinear regret in sublinear space. In particular, they show that under the standard oblivious adversary model, defined in Section \ref{sec:prelim}, $\tilde{O}(n^2 T^{\frac{2}{2+\delta}})$ regret (with high probability) is achievable in space $n^{\delta}$. 
Our work is also under this same general model studied in \cite{peng2023onlinesublinear} and we achieve regret $\tilde{O}(n^2 T^{\frac{1}{1+\delta}})$ in space $n^{\delta}$, improving upon the regret bound of \cite{peng2023onlinesublinear} for $\delta > 0$. In particular, as $\delta \rightarrow 1$ (ignoring $\text{poly}(n)$ factors), our regret approaches $\tilde{O}_n(\sqrt{T})$ but the regret of \cite{peng2023onlinesublinear} approaches $\tilde{O}_n(T^{2/3})$. The former is the right scaling as it is known that $\tilde{\Theta}(\sqrt{T})$ regret is achievable and necessary in the standard online learning with experts problem with no space considerations \cite{aroraHK12}.
\cite{peng2023onlinesublinear} also show that under the stronger adaptive adversary model where the input sequence can depend on the randomness used by the algorithm so far, sublinear regret in sublinear space is impossible. This additionally motivates the setting of oblivious adversaries.

Note that the lower bound given in \cite{srinivas2022memoryexperts} holds for the i.i.d. setting so it automatically also extends to the general setting of an oblivious adversary. However, as $\delta \rightarrow 0$, the lower bound of \cite{srinivas2022memoryexperts} implies that $\Omega(\sqrt{Tn})$ regret is required. If $n \ll T$, this implies that with logarithmic space, it might still be possible obtain sublinear (in $T$ regret). Either strengthening the lower bound to show that $\Omega(T)$ regret is required in this extremely small space regime, or strengthening our upper bounds in the very small memory regime are both exciting directions for future research.

\paragraph{The experts problem} 

Both the experts problem and the MWU algorithm are quite general and have found applications in many algorithmic and optimization problems including boosting, graph algorithms, portfolio optimization, linear programming, statistical estimation, and learning theory among many others \cite{brown1951iterative, littlestoneW89, freundS95, coverO96, ordentlichC98, bianchi06, christianoKMST11, garberH16, klivansM17, hopkins0Z20}. We especially refer to the survey \cite{aroraHK12} for further information. It is an interesting open direction to explore these applications in the bounded memory regime. 

\paragraph{Memory limited learning}Learning with constrained memory is a rich field of study in its own right with extensive works on convex optimization, kernel methods, statistical queries, and general machine learning algorithms \cite{williamsS00, rahimiR07, mitliagkasC013, shamir14, steinhardtD15, steinhardtVW16, kolRT17, raz17, daganS18, daganKS19, sharanSV19, gargRT19, bakshi22}. The focus of many of these works is on lower bounds or they are not in the online model. The work in these and related areas is extensive and we refer to the referenced works for further information.

Memory constraints have also been studied in the online multi-arm bandit settings, but these works are mostly focused on the stochastic settings and hence different than our worst-case viewpoint. We refer the readers to \cite{liauSPY18, chaudhuriK20, assadiW20, jinH0X21, maitiPK21} and the references within.

\section{Preliminaries}\label{sec:prelim}
We consider the learning with experts problem with $n$ experts over $T$ days. For every day $t \in [T]$, every expert $e$ makes a prediction $x_e^t \in [0,1]$. Then we receive feedback on the loss $\ell_t(e)$ of every expert $e$. We let $\ell_t \in [0, 1]^n$ denote the loss vector encoding the losses of all experts. We work under the exact same model as \cite{peng2023onlinesublinear}, which we summarize below.

\paragraph{Query Model}
In the space restricted model which we work under, we cannot explicitly store the losses and predictions of every expert at every time step. To formalize this, we first state how we obtain the losses of experts. Before every time step $t$, we are allowed to select $E_t$ experts and we observe their predictions. Then we pick a fixed expert $e \in E_t$ whose prediction we follow. Then the losses of all experts in $E_t$ is revealed and thus the loss of our algorithm on day $t$ is $\ell_t(e)$.

\paragraph{Memory/Space Model}
Our memory model can be understood in terms of the standard streaming model of computation. A word of memory holds $O(\log(nT))$ bits and at every time step $t$, our memory state consists of $M_t$ different words of memory. Some of these words are used to denote the experts $E_t$ whose predictions and losses we observe. After observing the losses, we update our memory state to use $M_{t+1}$ words of memory.  We say that an algorithm has space complexity $M$ if $M_t \le M$ for all $t \in [T]$. 

Our algorithm is randomized and assumes oracle access to random bits, similar to \cite{peng2023onlinesublinear}. This assumption can be easily removed by standard tools in the streaming literature via pseudorandom generators which only require additional poly-logarithmic space in $n$ and $T$ \cite{nisan92, indyk06}. Hence, we ignore the space complexity of generating random bits.

\paragraph{Adversary Model}
Our work is under an oblivious adversary, a standard assumption in streaming algorithms. The loss vectors $\ell_t$ for all $t \in [T]$ are  constructed by an adversary. The adversary can use randomness but their random bits are independent of any random bits used by our algorithm. Alternatively, the adversary first constructs the loss vectors in advance, and then afterwards our algorithm interacts with the adversary via the query model described above, using its own independent randomness.

\paragraph{Other Notation}
 We use $\tilde{O}$ notation to hide logarithmic terms in $n$ and $T$. As in prior works \cite{srinivas2022memoryexperts, peng2023onlinesublinear}, we assume that $n$ and $T$ are polynomially related for simplicity. We also note that our results extend easily to the case that $T$ is unknown in advance by applying the standard doubling tricks; see Remark \ref{rem:doubling}.
 
We also state the standard multiplicative weights update (MWU) algorithm, a key workhorse in our analysis, as well as its guarantees \cite{aroraHK12}.

\begin{lemma}[MWU guarantee, \cite{aroraHK12, peng2023onlinesublinear}]
\label{lem:mwu}
Suppose $n, T, \eta > 0$ and let $\ell_t \in [0,1]^n$ ($t\in [T]$) be the losses of all experts on day $t$. The multiplicative weight update algorithm satisfies
\begin{align*}
    \sum_{t=1}^{T} \langle p_t, \ell_t\rangle - \min_{i^{*}\in [n]}\ell_t(i^{*}) \leq \frac{\log n}{\eta} + \eta T,
\end{align*}
and with probability at least $1-\delta$,
\begin{align*}
    \sum_{t=1}^{T}  \ell_t(i_t) - \min_{i^{*}\in [n]}\ell_t(i^{*}) \leq \frac{\log n}{\eta} + \eta T + O\left(\sqrt{T \log (n/\delta)}\right).
\end{align*}

Taking $\eta = \sqrt{\frac{\log n}{T}}$, the MWU algorithm has a total regret of $O\left(\sqrt{T\log (nT)}\right)$ with probability at least $1- 1/\text{poly}(T)$ and a standard implementation takes $O(n)$ words of memory.
\end{lemma}

\section{Algorithm}
\label{sec:algo}

\begin{algorithm}[ht!]
\caption{\label{alg:predict} \predictalg}
\begin{flushleft}
{\bfseries Input:} Buckets $\B = B_0, \ldots, B_{k-1}$, weights across buckets $\w$\\
{\bfseries Output:} Array of $k$ predictions (experts) $\preds$
\begin{algorithmic}
\State Initialize $\preds$ to be an array of size $k$
\For{$i \in \{0, \ldots, k-1\}$}
    \State Pick an expert $e_i$ from $B_i$ proportional to their weights
    \If{$i = 0$}
        \State $\preds_0 \gets e_i$
    \Else
        \State $(w, w') \gets \w_{i-1}$ are the weights for the current and previous buckets, respectively
        \State $\preds_i \gets e_i$ with probability $\frac{w}{w + w'}$
        \State Otherwise, $\preds_i \gets \preds_{i-1}$
    \EndIf
\EndFor
\State \Return $\preds$
\end{algorithmic}
\end{flushleft}
\end{algorithm}

\begin{algorithm}[ht!]
\caption{\label{alg:weights} \weightsalg}
\begin{flushleft}
{\bfseries Input:} Buckets $\B = B_0, \ldots, B_{k-1}$, weights across buckets $\w$, algorithm predictions $\preds$, losses $\losses$, block sizes $\{T_i\}_{i=0}^k$ \\
{\bfseries Output:} Updated buckets $\B$ and weights $\w$
\begin{algorithmic}
\For{$i \in \{0, \ldots, k-1\}$}
    \State $\eta_i^{(1)} \gets \sqrt{\log(|B_i|)/T_i}$
    \State $\eta_i^{(2)} \gets \sqrt{2/T_i}$
    \For{expert $e \in B_i$} \Comment{Update internal weights within buckets}
        \State Increment loss for $e$ in $B_i$ by $\ell_t(e)$ 
        \State MWU of corresponding weight in $B_i$ with rate $\eta_i^{(1)}$ and loss $\ell_t(e)$
    \EndFor
    \If{$i > 0$} \Comment{Update external weights between buckets}
        \State MWU of big weight in $\w_{i-1}$ with rate $\eta_{i-1}^{(2)}$ and loss $\ell_t(\preds_i)$
    \EndIf
    \State MWU of small weight in $\w_i$ with rate $\eta_i^{(2)}$ and loss $\ell_t(\preds_i)$
\EndFor
\State \Return $\B, \w$
\end{algorithmic}
\end{flushleft}
\end{algorithm}

\begin{algorithm}[ht!]
\caption{\label{alg:buckets} \bucketsalg}
\begin{flushleft}
{\bfseries Input:} Buckets $\B = B_0, \ldots, B_{k-1}$, weights across buckets $\w$, block sizes $\{T_i\}_{i=0}^k$, space parameter $m$, time $t$ \\
{\bfseries Output:} Updated buckets $\B$ and weights $\w$
\begin{algorithmic}[1]
\State $\eps \gets \log (T)/m$
\State $k' \gets \max \left\{i \in \{0, \ldots, k-1\}: t \equiv 0 \pmod {T_i}\right\}$ \Comment{Number of buckets to update}
\State $\tau \gets \max\{t' \leq t : t \equiv 0 \pmod{T_{k'+1}}\}$ \Comment{Most recent time $B_{k'+1}$ was updated}
\For{$i \in \{0, \ldots, k'\}$} \Comment{Evict experts, reset weights}
    \State Remove any expert $e$ from $B_i$ with arrival time at least $\tau+0.5$ if $\exists e' \in B_i: e'$ $\eps$-dominates $e$ (according to \cref{def:dominance}) \label{line:eviction}
    \State Set the weights of all remaining experts in $B_i$ to $1$ 
    \State Set the weights in $\w_i$ to $1$
\EndFor \label{line:end-of-eviction}

\For{$i \in  \{k', k'-1, \ldots, 1\}$} \Comment{Forwarding experts from smaller to larger buckets}
        \State Create copies of all of the experts in buckets $B_0, \ldots, B_{i-1}$ in $B_i$ with unique arrival times in $[t,t + 0.5)$, losses $0$, and weights $1$
\EndFor \label{line:end-of-forward}

\For{$i \in \{0, \ldots, k'\}$} \Comment{Sample new experts}
    \State Sample $m$ experts (with replacement) and add them to $B_i$ with unique arrival times in $[t + 0.5, t+1)$, losses $0$, and weights $1$ \label{line:sample}
\EndFor
\State \Return $\B, \w$
\end{algorithmic}
\end{flushleft}
\end{algorithm}

\begin{algorithm}[ht!]
\caption{\label{alg:hierbaseline} \mainalg}
\begin{flushleft}
{\bfseries Input:} Block sizes $T_0 < T_1 < \ldots < T_k=T$ with $T_i \equiv 0 \pmod{T_{i-1}}$ for $i \in [k]$, number of experts $n$, space parameter $m$
\begin{algorithmic}[1]
\State Create $k$ buckets $\B = B_0, \ldots, B_{k-1}$ which will each store tuples of the form (expert, arrival time, loss, weight) \footnote{Multiple copies of the same expert with different arrival times, losses, and weights may be stored in the same bucket (these are treated as different entries in the bucket).}
\State Create array $\w$ of $k-1$ pairs of weights (big weight, small weight) initialized to $1$
\State $\B \gets \bucketsalg(\B, \w, \{T_i\}_{i=0}^k, m, 0)$ \Comment{Initialize buckets with samples}

\For{time $t$ in $\{1,\ldots, T\}$} 
    \State $\preds \gets \predictalg(\B, \w)$
    \State \textit{Play expert} $\preds_{k-1}$ \Comment{Make prediction}
    \State Observe losses $\losses \gets \{\ell_t(e): e \in B\}$ \Comment{Query loss vector}
    \State $\B, \w \gets \weightsalg(\B, \w, \preds, \losses, \{T\}_{i=0}^k)$ \Comment{Update weights}
    \If{$t \equiv 0 \pmod{T_0}$}
        \State $\B \gets \bucketsalg(\B, \w, \{T_i\}_{i=0}^k, m, t)$ \Comment{Update buckets}
    \EndIf
\EndFor
\end{algorithmic}
\end{flushleft}
\end{algorithm}

\subsection{Algorithm Description}
The ``Baseline'' algorithm of \cite{peng2023onlinesublinear} splits the total time $T$ into smaller periods (buckets) and tracks the progress of a small number of experts within each bucket. The predictions of the algorithm are formed by running MWU on the experts within each bucket (reset between buckets).
At the end of each bucket, experts that have not been performing well under some particular definition are evicted and some new experts are sampled (see \cref{sec:technical_overview} for a more detailed description of their algorithm).
Our algorithm uses this underlying framework but with a new eviction rule and a hierarchy of buckets of different sizes.

The main algorithm is \cref{alg:hierbaseline} with key subroutines in Algorithms \ref{alg:predict}, \ref{alg:weights}, and \ref{alg:buckets}.
The algorithm takes as input a set of $k+1$ values $T_0 < T_1 < \ldots < T_k = T$ with $T_i \equiv 0 \pmod{T_{i-1}}$ for $i \in [k]$ which specify the block sizes in the hierarchy of $k$ buckets.
In addition, the algorithm receives the number of experts $n$, and a space parameter $m$ (the space used by the algorithm will ultimately be $O(k^3 m)$).

Each bucket maintains a set of $O(k^3 m)$ experts (see \cref{lem:baseline-space}) along with an arrival time, loss, and weight for each expert. The arrival time is the time in $[T]$ at which the expert arrived in the bucket, the loss is the sum of the losses of the expert since that arrival time, and the weight corresponds to a MWU subroutine being run on the experts within the bucket.
At a high level, each timestep, each bucket within the hierarchy outputs a predicted expert to follow. Using MWU recursively between the hierarchies (using weights $\w$), we determine which predicted expert to output as our final prediction. We then observe the true losses for all experts in all buckets and update the multiplicative weights in MWU within each bucket as well as the multiplicative weights between buckets. Finally, if we have reached the end of a given bucket, we update the experts in that bucket by evicting some experts, resetting weights, forwarding experts in small buckets to large buckets, and sampling some new experts. We detail the procedures for \emph{prediction, updating weights, and updating buckets} below.

\paragraph{Prediction (\cref{alg:predict})} The prediction is essentially formed from two types of weights. Each bucket maintains a set of `internal' weights over its pool of experts. In addition, the algorithm maintains an `external' pair of weights over each consecutive pair of buckets. Each bucket picks an expert $e_i$ from its pool proportional to its internal weights. We say the prediction at the smallest bucket is $y_0 = e_0$. Then, recursively from smaller to larger buckets, the prediction at the $i$th step $y_i$ is picked between the prediction of $e_i$ and $y_{i-1}$ proportional to the corresponding external weights.

\paragraph{Updating weights (\cref{alg:weights})}
For each bucket, we observe the losses of each expert in the bucket. We update the corresponding internal and external weights according to the multiplicative weight update rule.

\paragraph{Updating buckets (\cref{alg:buckets})}
If we reach a time where $t \equiv 0 \pmod{T_i}$, we update the entries in bucket $B_i$. We let $k'$ be the index of the largest bucket we update and let $\tau$ be the last time the bucket $B_{k'+1}$ was updated.
Updating buckets is comprised of four steps: eviction, resetting weights, forwarding, and sampling.
First, for each bucket to be updated, we evict all experts in the pool which have arrived since time $\tau + 0.5$ using the following dominance definition.
\begin{definition}[$\eps$-Dominance]\label{def:dominance}
Each expert $e$ tracked by the algorithm will have an arrival time and loss since that time. Given a parameter $\eps$, an expert $e'$ $\eps$-dominates $e$ if $e'$ arrives before $e$  and the loss of $e'$ is at most $(1+\eps)$ times the loss of $e$.
\end{definition}
Of the experts up for eviction, if they are dominated by any other expert in the pool using $\eps = \log(T)/m$, they are evicted.
Then, we reset the internal weights of all remaining experts in the pool as well as the external weights between buckets levels $i-1$ and $i$ to $1$.

Next, for each bucket that needs to be updated, we make copies of all the experts remaining after eviction from smaller buckets. We copy the experts but do not copy their arrival times or losses. Instead, we set their arrival times to the current time and set their losses to $0$.
Finally, we sample $m$ new experts. Both for forwarding and sampling of experts, we give each expert a unique arrival time in $[t,t+1)$ in order for the dominance rule to be well-defined.

\subsection{Analysis of improved baseline algorithm}
In this section, we analyze the improved baseline algorithm, \cref{alg:hierbaseline}, proving upper bounds on both its space usage and its regret. The following lemma bounds the space usage.

\begin{lemma}\label{lem:baseline-space}
\cref{alg:hierbaseline} uses $O(k^3 m)$ words of space.
\end{lemma}

\begin{proof}
The space required by the algorithm is dominated by the number of entries in the buckets $\B$. For each entry in the bucket, we store a constant number of words, and in total, we query a number of experts at each timestep which is at most the number of entries in $\B$. Thus the number of words stored at any timestep $t$ can be uppper bounded by a fixed constant factor times the number of entries in $\B$.

Consider any timestep $t$ and any bucket $B_i$ such that $t \equiv 0 \pmod{T_i}$, i.e., $B_i$ will be updated in \cref{alg:buckets}. Recall that $\tau$ is the start time of the smallest bucket which is not being updated. We will argue that after the eviction step on line \ref{line:end-of-eviction}, $B_i$ will contain at most $O(m)$ experts that have arrived since time $\tau + 0.5$. 
Since all experts have unique arrival times and by \cref{def:dominance} of $\eps$-Dominance, every expert $e$ remaining after eviction must have less than a $1/(1+\eps)$ fraction of the loss of any remaining expert that arrived before $e$. As the losses are in $[0,1]$, there can be at most $\log_{1+\eps}(T) = O(\log(T)/\eps) = O(m)$ remaining experts out of the pool of experts up for eviction (those that have arrived since $\tau + 0.5$).

For each $j > i$, $B_i$ may contain up to $O(m)$ experts which survived the last round of eviction during which $k' = j$ ($B_j$ was the largest bucket to be updated) via the argument above. Any other experts are considered for eviction in the current step and so the total number of experts after eviction in $B_i$ is at most $O((k-i)m)$.

After forwarding on line \ref{line:end-of-forward}, $B_i$ has received experts which survived eviction from all $B_j$ for $j < i$. In total, the number of experts in $B_i$ after this step is
\[
    O\left(\sum_{j = 0}^{i} (k-j)m \right) = O(kmi).
\]
Sampling on line \ref{line:sample} adds $m$ extra entries.
The total number of entries in all buckets is therefore
\[
    O\left(\sum_{i=0}^{k-1} kmi \right) = O(k^3m). \qedhere
\]
\end{proof}

We next proceed to analyze the regret of the baseline algorithm. Our goal is to prove the following theorem.
\begin{theorem}\label{thm:baseline-bound}
For any constant $\gamma>0$, with probability $1-O(T^{-\gamma})$, the regret of~\cref{alg:hierbaseline} is
    \begin{align}
        \tilde O\left(\frac{T_k}{m} + \frac{T_k}{\sqrt{T_{k-1}}}+\frac{n}{m}\left(\sum_{i=0}^{k-1} \frac{T_{i}}{\sqrt{T_{i-1}}}\right)\right),
    \end{align}
    where the $\tilde O$ hides logarithmic factors in $T$. Here, we have defined $T_{-1}=1$ for convenience.
\end{theorem}
Since we assume that the adversary is oblivious, it suffices to fix any loss sequence $\ell_1,\dots,\ell_T\in [0,1]^n$ and show that with high probability, the algorithm achieves the above regret bound on this loss sequence. We will first set up some notation. Then we introduce the terminology of \emph{\stay blocks} and \emph{\evict blocks} and state a series of technical lemmas on the structure of these. Finally, we demonstrate how the theorem follows from these lemmas. 

\paragraph{Notation}
We denote the optimal expert over the $T$ days with loss sequence $\ell_1,\dots,\ell_T\in [0,1]^n$ by $e^*$. For an expert $e$ and an interval $I\subseteq [T]$, we denote by
\[
L_{e}(I)=\sum_{t\in I}\ell_t(e)
\]
the loss of expert $e$ over time $I$. For two experts $e$ and $e'$, we denote by
\[
R_{e,e'}(I)=L_{e}(I)-L_{e'}(I)
\]
the regret of $e$ relative to expert $e'$ over time $I$. We will solely apply this definition with $e'=e^*$. More generally, for two algorithms $\mathcal{A}$ and $\mathcal{A}'$ which plays sequences $(i_1,\dots,i_T)$ and $(i_1',\dots,i_T')$, and an interval $I\subseteq [T]$, we denote by
\[
R_{\mathcal{A},\mathcal{A}'}(I)=\sum_{t\in I}\ell_t(i_t)-\sum_{t\in I} \ell_t(i_t'),
\]
which we refer to  as the the regret of algorithm $\mathcal{A}$ relative to algorithm $\mathcal{A}'$ over time $I$.

\paragraph{Definition of \stay blocks and \evict blocks} 
\begin{figure}[ht]
    \centering
    \includegraphics[width=\textwidth]{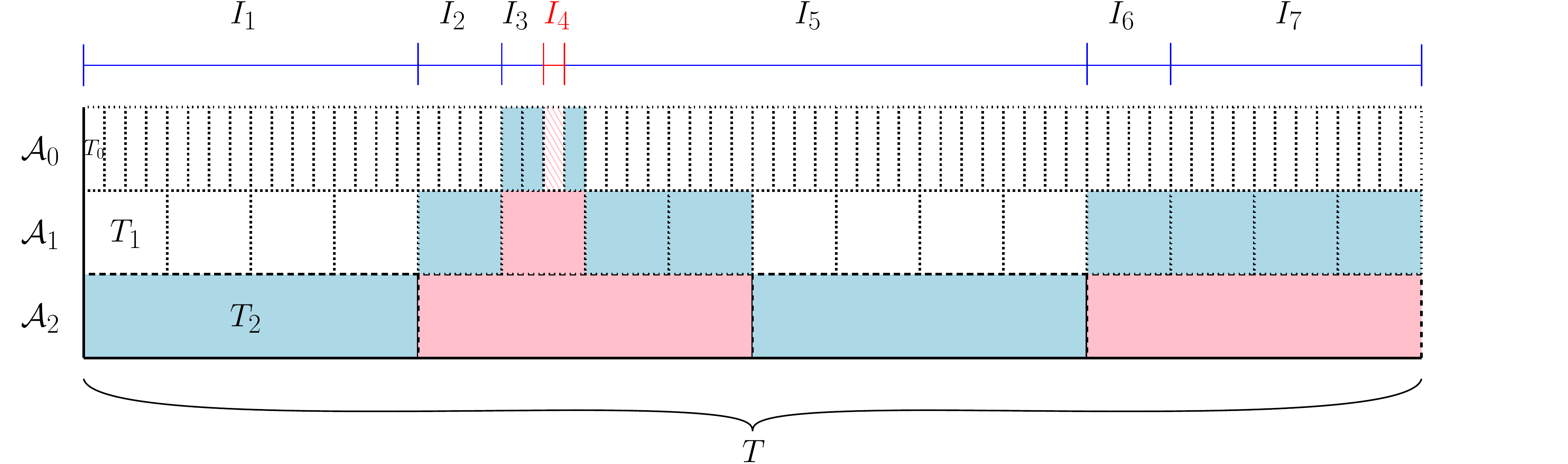}
    \caption{An instance of our algorithm with $k=3$. At level $i$, time is partitioned into intervals of $T_i$ days. The red blocks in the figure are the blocks appearing in $\mathcal{S}$. The intervals on top are the intervals defined during the construction of $\mathcal{S}$.  Each interval $I_j$ corresponds to an expert $e_j$ with low regret compared to the optimal expert $e^*$ (except for the \emph{marked} interval $I_4$ which corresponds to any expert present in the corresponding block at level $0$). 
    For an interval $I_j$ with corresponding expert $e_j$, a blue block underneath $I_j$ at some level $i$ indicates that expert $e_j$ is present in the pool of level $i$ during the days of the block.}
    \label{fig:redbluediagram}
\end{figure}
Recall that~\cref{alg:hierbaseline} uses $k$ recursive levels, denoted levels $0,\dots,k-1$. At level $i$, the $T=T_k$ days are partitioned into blocks of $T_i$ days where $T_0<\cdots<T_k$. To be precise, for a given level $i$ and $j\in [T/T_i]$, we define $B_{i,j}=\{(j-1) T_i+k\mid 1\leq k\leq T_i\}$ which is the $j$'th block of days at level $i$. We define $P'_{i,j}$ to be the experts present in the pool of level $i$ immediately prior to the days in block $B_{i,j}$ after the eviction process and the forwarding of experts between levels but before the sampling of the $m$ new experts. Let us further define $P_{i,j}$ to be the level $i$ pool obtained after sampling the $m$ new experts and including them in $P_{i,j}'$. We give these new experts distinct arrival times such that the eviction rule is well-defined among them. For the analysis, we will simply think of this sampling as being done one expert at a time\footnote{Recall that we sample with replacement from the set of all experts. If an expert which is already contained in the pool gets sampled, we include a new copy of this expert in the pool. At any point, we can therefore have multiple copies of an expert in the pool, but they will be assigned different arrival times times and loses since their arrival. For the evictions, we view these experts as being different. In particular, an expert can even get evicted by a copy of itself that was sampled earlier in time.}. Let $P_{i,j}'=P_{i,j}^{(0)}\subset P_{i,j}^{(1)}\subset \cdots P_{i,j}^{(m)}=P_{i,j}$ be the (multi-)sets of experts in the level $i$ pool after including each new sampled expert in their given arrival order. 

We start by describing a partition of the blocks $B_{i,j}$ into two types, namely \emph{\stay-blocks} and \emph{\evict-blocks}. This partitioning will only depend on the random sampling of experts in the $k$ levels of our algorithm, and not on the outcomes of the various MWU algorithms that we run as subroutines. Consider any fixed level $i$, any $j\in[T/T_i]$ and the pool of experts $P_{i,j}'$. Before sampling each of the $m$ new experts to be added to the level $i$ pool $P_{i,j}'$, we can ask the following hypothetical question: \emph{If we sample the globally optimal expert $e^*$, would it stay in the pool at level $i$ until the end of time, or would it be evicted by some expert $e$}? Since the loss sequence is fixed in advance and since an expert can only be evicted by another expert with an earlier arrival time, the answer to this question only depends on the set of experts in the pool immediately before the sampling of the new expert.  

We will say that the block $B_{i,j}$ is a \emph{\stay-block}, if for each $0\leq t\leq m-1$, it holds that \emph{if} $e^*$ is sampled as the next expert to form $P_{i,j}^{(\ell+1)}$ 
from $P_{i,j}^{(\ell)}$, then $e^*$ will stay in the level $i$ pool until the end of time. If a block is a \stay-block and the optimal expert $e^*$ is in fact sampled (and thus will stay to the end), we say that the \stay-block is \emph{actualized}.
If $B_{i,j}$ is not a \stay-block, we say that it is an \emph{\evict-block}. It is clear that whether or not a given block is a  \stay-block or an \evict-block, depends only on the random sampling of experts. Importantly, we note that an \evict-block $B_{i,j}$ at level $i$ gives rise to an expert $e$ and an interval of days $I=\{(j-1)T_i+1,(j-1)T_i+2,\dots\}$ such that $L_e(I)\leq (1+\eps) L_{e^*}(I)$ and expert $e$ exists in the level $i$ pool during the days in $I$. Indeed, we can let $e$ be any expert in the pool that would evict $e^*$ and let $I$ denote the sequence of days until this eviction would happen. Then by the eviction rule, $L_e(I)\leq (1+\eps) L_{e^*}(I)$. 
\\
\\
We will next specify a subset $\mathcal{S}$ of the \stay-blocks. While defining the blocks in $\mathcal{S}$, we will further define certain intervals of time $I_1,\dots, I_r$ and experts $e_1,\dots, e_r$. We will further \emph{mark} some of these intervals. All of these quantities are random variables depending only on the random sampling of experts that the algorithm performs in its $k$ layers. To define them, we use the following inductive procedure illustrated in~\cref{fig:redbluediagram}: Initially, we set $t=1$ 
and let $\ell=k-1$.
Suppose recursively that for some $r_0\geq 1$ we have defined $I_1,\dots,I_{r_0-1}$ and $e_1,\dots, e_{r_0-1}$. Suppose further that we are given $t$ and $\ell$ in a way such that at time $t$ a new block of days begins at level $\ell$. Denote this block $B(t,\ell)=\{t,t+1,\dots, t+T_\ell-1\}$. We consider two cases.
\begin{itemize}
\item {\bf $B(t,\ell)$ is an \evict-block.} By the observations above, there exists an interval of days $I$ containing $B(t,\ell)$, and an expert $e$ which is contained in the level $\ell$ pool during the days in $I$ such that $L_e(I)\leq (1+\eps) L_{e^*}(I)$. We define $I_{r_0}:= I$ and $e_{r_0}:= e$. We further update $t\leftarrow t+T_\ell$. Finally, we update $\ell$ to be maximal in $\{0,1,\dots, k-1\}$ such that a new block begins at level $\ell$ at time $t$.
\item {\bf $B(t,\ell)$ is a \stay-block.} We include $B(t,\ell)$ in $\mathcal{S}$, and split into further cases:
\begin{enumerate}
\item \emph{$B(t,\ell)$ is actualized:} This means that $e^*$ gets sampled to the level $\ell$ pool at time $t$ and consequently will stay to the end of time (since $B(t,\ell)$ is a \stay-block). We set $e=e^*$ and $I=[t,t+1,\dots, T]$ and terminate the procedure.
\item \emph{$B(t,\ell)$ is not actualized:} If $\ell=0$, we define $I_{r_0}=B(t,0)$ and let $e_{r_0}=e$ to be an arbitrary expert in the level $0$ pool during the days $B(t,0)$. We further \emph{mark} the interval $I_{r_0}$. We update $t\leftarrow t+T_0$. Finally, we update $\ell$ to be maximal in $\{0,1,\dots, k-1\}$ such that a new block begins at level $\ell$ at time $t$. If on the other hand $\ell\neq 0$, we keep $t$ unchanged, update $\ell\leftarrow \ell-1$ and proceed to the next step.
\end{enumerate}
\end{itemize}
We will next prove a series of technical lemmas regarding the structure of the \stay blocks in $\mathcal{S}$, on the intervals $I_1,\dots,I_r$, and on the experts $e_1,\dots,e_r$. We first provide a high probability upper bound on  the size of $\mathcal{S}$.

\begin{lemma}\label{lemma:bound-on-yes}
For any $\delta \in (0,1)$, with probability at least $1-\delta$, it holds $|\mathcal{S}|\leq \frac{n}{m}\log(1/\delta)$.
\end{lemma}
\begin{proof}
Each time we encounter a \stay-block in the procedure above, we get $m$ chances to actually sample the optimal expert $e^*$. If we do indeed sample $e^*$, the procedure ends and we add no more blocks to $\mathcal{S}$. Since for each sample, the probability of sampling $e^*$ is $1/n$, we get that 
\[
\Pr[|\mathcal{S}|\geq N+1]\leq (1-1/n)^{Nm}\leq \exp\left(-\frac{Nm}{n}\right).
\]
Choosing $N=\frac{n}{m}\log(1/\delta)$, we obtain the desired result. 
\end{proof}
For the next lemma, we will consider the \emph{synthetic} algorithm $\mathcal{A}^{synth}$ which for $i=1,\dots, r$ plays expert $e_i$ during the days of $I_i$. This algorithm cannot be implemented as it would require knowledge of the future performance of experts, but we will use it merely for the analysis.
\begin{lemma}\label{lemma:regret-of-A0}
For any $\delta \in (0,1)$, with probability at least $1-\delta$ (over the random sampling of experts done by~\cref{alg:hierbaseline}), algorithm $\mathcal{A}^{synth}$ has regret at most $\eps T+\frac{n}{m}T_0\log(1/\delta)$.
\end{lemma}
\begin{proof}
We note that the intervals $I_1,\dots,I_r$ partition $[T]$. During any unmarked interval $I_{j}$, algorithm $\mathcal{A}^{synth}$ plays and expert $e_j$ such that $L_{e_j}(I_j)\leq (1+\eps) L_{e^*}(I_j)$. During a marked interval $I_{j}$, we use the trivial upper bound $L_{e_j}(I_j)\leq |T_0|$. This follows since only intervals at level $0$ are marked. By~\cref{lemma:bound-on-yes}, the number of marked intervals is at most $\frac{n}{m}\log(1/\delta)$ with probability at least $1-\delta$. Let $J_0=\{j\in [r] \mid \text{ $j$ is unmarked}\}$ and $J_1=[r]\setminus J_0$. Note that for each $j\in J_0$,
\[
R_{e_j,e^*}(I_j) = L_{e_j}(I_j)-L_{e^*}(I_j)\leq \eps L_{e^*}(I_j)\leq \eps|I_j|.
\]
Thus, with probability at least $1-\delta$, the regret of algorithm $\mathcal{A}^{synth}$ is at most
\[
|J_1| |T_0|+\sum_{j\in J_0}R_{e_j,e^*}(I_j)\leq \frac{n}{m}T_0\log(1/\delta)+\sum_{j\in J_0}\eps|I_j|\leq \frac{n}{m}T_0\log(1/\delta)+\eps T, 
\]
as desired.
\end{proof}
We are now ready to prove~\cref{thm:baseline-bound}.
\begin{proof}[Proof of~\cref{thm:baseline-bound}]
We first remark that the sampling, forwarding, and eviction of experts is completely independent of the random bits used to perform our various MWUs. Indeed, whether an expert is evicted or forwarded depends only on its loss compared to other experts since the time it was sampled. We may therefore \emph{fix} a sampling and then bound the regret solely over the randomness of the multiplicative weights updates. To be precise, we fix a sampling of experts such that $|\mathcal{S'}|=O(\frac{n}{m}\log T)$ which by~\cref{lemma:bound-on-yes} holds with high probability in $T$.

Let us briefly recall how the MWU of~\cref{alg:hierbaseline} works and introduce some notation for the proof. First, for each layer $i$, we run MWU on the pool of experts over each block of $T_i$ days. Let us call these algorithms $\mathcal{A}_0',\dots, \mathcal{A}_{k-1}'$ labelled according to level. We define $\mathcal{A}_0=\mathcal{A}_0'$ and recursively let $\mathcal{A}_i$ be the algorithm obtained by running MWU between $\mathcal{A}_{i}'$ and $\mathcal{A}_{i-1}$ restarting after every block of length $T_i$. Then~\cref{alg:hierbaseline} is simply the algorithm $\mathcal{A}_{k-1}$.
Let us define $\mathcal{S}_i=\{B\in \mathcal{S}\mid |B|=|T_i|\}$, namely the blocks of $\mathcal{S}$ at level $i$. Let us further define the subsets of time $J_0\subseteq J_1\subseteq\cdots \subseteq J_k=[T]$ as follows: $J_k=T$ and for $0\leq i\leq k-1$, we set $J_i=\bigcup_{B\in \mathcal{S}_i}B$.
We will prove the following claim by induction:
\begin{claim}\label{claim:regret-comp}
For any constant $\gamma>0$, there exists a constant $C>0$ such that with probability $1-O(T^{-\gamma})$ for each $0\leq i \leq k-1$, 
\begin{align}\label{eq:tech1}
R_{\mathcal{A}_i,\mathcal{A}^{synth}}(J_i)\leq C\sqrt{\log n T} \cdot \sum_{j=0}^{i} \frac{|\mathcal{S}_j|T_j}{\sqrt{T_{j-1}}}.
\end{align}
Here we have defined $T_{-1}=1$ for convenience. Moreover, 
\begin{align}\label{eq:tech2}
R_{\mathcal{A}_{k-1},\mathcal{A}^{synth}}(T)\leq C\sqrt{\log n T} \cdot \left(\frac{T}{\sqrt{T_{k-1}}}+\sum_{j=0}^{k-1} \frac{|\mathcal{S}_j|T_j}{\sqrt{T_{j-1}}}\right).
\end{align}
\end{claim}
\begin{proof}[Proof of Claim]
We first fix $C$ sufficiently large such that with probability $1-O(T^{-\gamma})$, all the different MWUs that we run over any amount of days $T'$ have regret at most $(C/2)\sqrt{\log (nT)}\sqrt{T'}$ compared with the best fixed option. This is possible by~\cref{lem:mwu}. 

We prove the claim by induction on $i$. For $i=0$ it holds trivially as \[R_{\mathcal{A}_0,\mathcal{A}^{synth}}(J_0)\leq |J_0|=\frac{|\mathcal{S}_0|T_0}{\sqrt{T_{-1}}}.\] Now let $i\in \{1,\dots, k-1\}$ be given and assume inductively that the bound holds for smaller values of $i$. We claim that due to the forwarding of experts between levels of the algorithm and eviction rule, for each block of $T_{i-1}$ days $B\subseteq J_i\setminus J_{i-1}$ at level $i-1$, the corresponding level $i-1$ pool contains the expert played by $\mathcal{A}^{synth}$ during the days of $B$. 
To see why this is the case, note that any interval $I_j$ which contains the leftmost day of $B$ will in fact contain all the days of $B$ due to our eviction rule, namely, if expert $e_j$ survives to the beginning of block $B$, it can be evicted by the end of block $B$ at the earliest. Moreover, due to the forwarding, $e_j$ will indeed be present in the pool at level $i-1$ during the days of block $B$. By the guarantees of MWU, it follows that 
\[
R_{\mathcal{A}_{i-1},\mathcal{A}^{synth}}(J_i\setminus J_{i-1})\leq (C/2)\sqrt{T_{i-1}\log (nT)}\cdot \frac{|J_i\setminus J_{i-1}|}{T_{i-1}}\leq (C/2)\sqrt{\log n T}\frac{|\mathcal{S}_i| T_i}{\sqrt{T_{i-1}}}.
\]
Combining this with the induction hypothesis, 
\[
R_{\mathcal{A}_{i-1},\mathcal{A}^{synth}}(J_i)\leq (C/2)\sqrt{\log n T}\frac{|\mathcal{S}_i| T_i}{\sqrt{T_{i-1}}}+ C\sqrt{\log n T} \cdot \sum_{j=0}^{i-1} \frac{|\mathcal{S}_j|T_j}{\sqrt{T_{j-1}}}.
\]
Finally, by the guarantee of MWU between $\mathcal{A}_i'$ and $\mathcal{A}_{i-1}$ (which gets restarted after every block of length $T_i$), 
\begin{align*}
R_{\mathcal{A}_{i},\mathcal{A}^{synth}}(J_i) \leq & R_{\mathcal{A}_{i-1},\mathcal{A}^{synth}}(J_i)+(C/2)\sqrt{\log n T}|\mathcal{S}_i|\sqrt{T_i} \\
\leq & R_{\mathcal{A}_{i-1},\mathcal{A}^{synth}}(J_i)+(C/2)\sqrt{\log n T}\frac{|\mathcal{S}_i| T_i}{\sqrt{T_{i-1}}} \\
\leq & C\sqrt{\log n T} \cdot \sum_{j=0}^{i} \frac{|\mathcal{S}_j|T_j}{\sqrt{T_{j-1}}}
\end{align*}
where the second to last inequality used that $|T_i|>|T_{i-1}|$. This is the desired bound. Finally,~\eqref{eq:tech2} essentially follows from the same argument. Due to the forwarding of experts, for each block of days $B\subseteq T\setminus J_{k-1}$ at level $k-1$, the pool at level $k-1$ contains the expert played by $\mathcal{A}^{synth}$ during the days of $B$. Thus, 
\[
R_{\mathcal{A}_{k-1},\mathcal{A}^{synth}}(T)\leq  C\sqrt{\log n T} \cdot \frac{T}{\sqrt{T_{k-1}}}+ R_{\mathcal{A}_{k-1},\mathcal{A}^{synth}}(J_{k-1}),
\]
where the first term comes from the restarting of MWU between $T_{k-1}$ blocks and the second term is the regret of $\mathcal{A}_{k-1}$ in $J_{k-1}$. Plugging in the bound of~\eqref{eq:tech1} with $i=k-1$ gives the desired result. 
\end{proof}
To finish the proof from the claim, we simply upper bound each $|\mathcal{S}_i|\leq |\mathcal{S'}|=O(\frac{n}{m}\log T)$. Thus, we get that 
\[
R_{\mathcal{A}_{k-1},\mathcal{A}^{synth}}(T)=O\left(\log (Tn)^2\cdot \left(\frac{T}{\sqrt{T_{k-1}}}+\frac{n}{m}\sum_{j=0}^{k-1}\frac{T_j}{\sqrt{T_{j-1}}}\right)\right).
\]
Combining this with~\cref{lemma:regret-of-A0} (with $1/\delta$ a sufficiently high degree polynomial in $T$), we obtain the desired regret bound.
\end{proof}

\subsubsection{Setting the parameters}
In this section, we set the parameters in~\cref{thm:baseline-bound}, to obtain the following theorem.
\begin{theorem}\label{thm:final-baseline-bound}
Assume $T\geq C\frac{n}{m}$ for a sufficiently large constant $C$. There exists a choice of parameters for~\cref{alg:hierbaseline} such that with high probability in $T$, its regret is
\begin{align}
\tilde O \left(\frac{T}{m}+\sqrt{\frac{Tn}{m}}\right).
\end{align}
\end{theorem}
\begin{proof}
We start by fixing the values of the $T_i$'s. Let's $T_{-1}=1$ and $T_0=(T\frac{m}{n})^{\frac{2^k}{2^{k+1}-1}}$. Note that $T_0\geq \sqrt{C}> T_{-1}$ by the assumption on $T$. Now inductively, for $1\leq i \leq k-1$, we set $T_i=\frac{T_{i-1}^{3/2}}{\sqrt{T_{i-2}}}$. We remark that this choice of parameters solves the set of equations
\[
\frac{T}{\sqrt{T_{k-1}}}=\frac{n}{m}\frac{T_{k-1}}{\sqrt{T_{k-2}}}=\frac{n}{m}\frac{T_{k-2}}{\sqrt{T_{k-3}}}=\cdots =\frac{n}{m}\frac{T_{1}}{\sqrt{T_{0}}}=\frac{n}{m}T_0.
\]
It follows by a simple inductive argument that for each $i$, $T_i>T_{i-1}$. In fact, another straightforward inductive argument shows that for $0\leq i \leq k-1$, we have the equation $T_i=T_0^{\frac{2^{i+1}-1}{2^i}}$. Note that $T_{k-1}=\left(\frac{Tm}{n}\right)^{\frac{2^{k+1}-2}{2^{k+1}-1}} < T$, so this is a legal setting of parameters.  Now it is clear that with this choice of parameters, $\frac{T_i}{\sqrt{T_{i-1}}}=\frac{T_{i-1}}{\sqrt{T_{i-2}}}$ for $i=1,\dots,k-1$. Moreover, it is easy to check that $\frac{T}{\sqrt{T_{k-1}}}=\frac{n}{m}\frac{T_{k-1}}{\sqrt{T_{k-2}}}$. Thus, it follows from~\cref{thm:baseline-bound} that the regret is bounded by 
\[
\tilde O\left(\frac{T}{m}+\frac{kn}{m}T_0 \right).
\]
Note that for each $1\leq i \leq k-1$,
\[
\frac{T_i}{T_{i-1}}=T_0^{1/2^i}\geq T_0^{1/2^{k}}\geq \left(T\frac{m}{n}\right)^{\frac{1}{2^{k+1}}},
\]
Thus, choosing $k=\lg\lg (Tm/n)-O(1)$, where the $O(1)$ hides a sufficiently large constant, the gaps between the block sizes can be made larger than any constant. With this choice of $k$, 
\[
T_0=\left(\frac{Tm}{n}\right)^{1/2+\frac{1}{2^{k+1}-1}}=O\left(\sqrt{\frac{Tm}{n}}\right),
\]
so, we obtain the final regret bound of 
\[
\tilde O\left(\frac{T}{m}+k\sqrt{\frac{Tn}{m}}\right)=\tilde O\left( \frac{T}{m}+\sqrt{\frac{Tn}{m}} \right),
\]
as desired.
\end{proof}

\section{Bootstrapping}
\label{sec:boostrap}
In this section we describe the bootstrapping technique of the previous
work. For convenience, we restate their result using the following
lemma.

\bootstrap*
\begin{proof}
Let $j^{*}$ be the index of the best expert. The algorithm operates
over $T_{1}$ episodes of $T_{2}$ days each. Suppose we run a copy
of ALG2 that resets for every episode. For each episode $i$ and original
expert $e_{j}$, let synthetic expert $s_{j}$ be the answer of the
MWU with two choices, the original expert $e_{j}$ and the copy of
ALG2. By the guarantee of ALG2, the copy of ALG2 has regret at most
$R_{2}$ during episode $i$. By the union bound, with probability
$1-\delta(nT_{1}+1)$, we assume that all MWU applications and the
algorithms succeed. By the guarantee of MWU, $e_{i,j}$ has regret
at most $R_{3}:=O\left(\sqrt{T_{2}\ln\left(n/\delta\right)}\right)$
higher than ALG2. Next, consider a new expert problem defined over
the episodes as follows. Let $loss_{i}\left(e\right)$ be the loss
of expert $e$ for episode $i$ and the truncated loss 

\[
tloss_{i}\left(e\right):=\max\left(loss_{i}\left(e\right)-loss_{i}\left(ALG2\right),-R_{2}\right)
\]

We note two important properties of the truncated loss. First, notice
that $\left|tloss_{i}\left(s_{j}\right)\right|\le\max\left(R_{2},R_{3}\right)\ \forall j$.
Second, for the best expert $e_{j^{*}}$, we have 
\[
tloss_{i}\left(s_{j^{*}}\right)+loss_{i}(ALG2)=\max\left(loss_{i}\left(s_{j^{*}}\right),loss_{i}\left(ALG2\right)-R_{2}\right)\le loss_{i}\left(e_{j^{*}}\right)+R_{3}.
\]
We run ALG1 over the episodes for the $n$ synthetic experts using
the truncated loss. By the guarantee of ALG1, the regret of the algorithm
compared with the best synthetic expert is at most $R_{1}\cdot\max\left(R_{2},R_{3}\right)$.
We have

\begin{align*}
\sum_{i}loss_{i}\left(ALG1\right) & \le\sum_{i}tloss_{i}\left(ALG1\right)+loss_{i}\left(ALG2\right)\\
 & \le R_{1}\cdot\max\left(R_{2},R_{3}\right)+\sum_{i}tloss_{i}\left(s_{j^{*}}\right)+loss_{i}\left(ALG2\right)\\
 & \le R_{1}\cdot\max\left(R_{2},R_{3}\right)+\sum_{i}\left(loss_{i}\left(e_{j^{*}}\right)+R_{3}\right).
\end{align*}
Thus, the total regret is at most $O\left(T_{1}R_{3}\right)+R_{1}R_{2}$.

Note that the space complexity also increases by a constant multiplicative factor. This is because we only have to create synthetic experts for the experts actually used by ALG1. Since they both have space complexity at most $m$, ALG1 will only query at most $m$ experts at any day (see our space model in Section \ref{sec:prelim}). Thus we can charge the extra parameters we need to keep track of in the MWU in the creation of the synthetic experts to the experts tracked by ALG1. Likewise, any copy of ALG2 also requires $O(m)$ additional space, but we only have one copy of ALG2 at any day.
\end{proof}
\begin{corollary}\label{cor:bootstrap}
Suppose there is an algorithm which over $T$ days achieves regret $R$ with probability at least $1-\delta$ and uses space $m$. Then for $i\geq 1$ there exists an algorithm which over $T^i$ days achieves regret
\[O\left(\left(R^{i}+\sum_{j=0}^{i-2}R^{j}T^{(i+1-j)/2}\right)\sqrt{\ln\left(n/\delta\right)}\right),\]
with probability at least $1-\delta(nT+1)^{i-1}$.
Additionally, there exists a constant $C$ such that the space used by this is at most $C^i m$.
\end{corollary}

\begin{proof}
The corollary immediately follows by induction. The case $i=1$ follows by assumption. For $i>1$, we set $T_1 = T, R_1 = R$ and $T_2 = T^{i-1}$ and 
\[
    R_2 = \left(R^{i-1} + \sum_{j = 0}^{i-3}R^j T^{(i-j)/2}\right) \sqrt{\ln(n/\delta})
\]
and plug into the lemma.
\end{proof}

\section{Putting Everything Together}\label{sec:main_proof}
We are now ready to put the results of the previous section together to prove the main theorem of the paper.

\mainthm*

\begin{proof}
Assume first that $T$ is of the form $(nm)^i$ for some integer $i$. Note that $i=O(1)$. By~\cref{thm:final-baseline-bound}, there exists an algorithm which over a sequence of $T_0=nm$ days achieves regret $\tilde O(n)$ with high probability in $T$. Plugging this algorithm into~\cref{cor:bootstrap} as the base algorithm, there exists an algorithm which over $T$ days achieve regret
\[
R=\tilde O\left(n^i+\sum_{j=0}^{i-2}n^j(nm)^{\frac{i+1-j}{2}}\right)=\tilde O(n^{i+1})
\]
also with high probability in $T$. (Note that the polylogarithmic factors also blow up by a constant factor of $i$, which is hidden in the $\tilde{O}$ notation.) Here we simply bounded $m\leq n$ for all occurrences of $m$. Using that $T=(nm)^i$, and putting $m=n^\delta$, we obtain that 
\[
R=\tilde O \left(nT^{\frac{1}{1+\delta}
}\right),
\]
as desired.

If $T$ is not of the form $(nm)^i$, we pick $i\in \N$ minimal such that $T'=(nm)^i\geq T$ and artificially extend the sequence of loss vectors by giving all experts the same loss for the remaining $T'-T$ days. We then run the algorithm with parameters corresponding to $T'$ but terminate after $T$ days. The regret of the algorithm on the extended sequence is the same as that of the original sequence and is bounded by
\[
\tilde O \left(nT'^{\frac{1}{1+\delta}}\right)=\tilde O \left(n(nmT)^{\frac{1}{1+\delta}} \right)=\tilde O \left(n^2 T^{\frac{1}{1+\delta}} \right),
\]
as desired. The space used by the algorithm can be bounded  using~\cref{lem:baseline-space}. Note that as $i$ is a constant, the bootstrapping procedure increases the space by only a constant factor by \cref{cor:bootstrap}. Recall that the baseline algorithm used $k=O(\log \log n)$ recursive layers, and thus, by the lemma, the space usage can be bounded by $O(k^3m)=\tilde O(n^\delta)$.
\end{proof}
\begin{remark}\label{rem:doubling}
As in prior works \cite{peng2023onlinesublinear}, our results can be easily made to handle the case where $T$ is unknown: we simply guess the value of $T$ as done in the proof of Theorem \ref{thm:main}. In more detail, we run multiple copies of our algorithm with geometrically increasing guesses of $T$ and run MWU over all the different algorithm copies. This only introduces an additional $O(\log n)$ factor blow up in the space and regret bounds.
\end{remark}

\bibliographystyle{alpha}
\bibliography{bib}

\end{document}